\documentclass[11pt,reqno]{amsart}
\usepackage{amscd,amssymb,amsmath,amsthm}
\usepackage[arrow,matrix]{xy}
\usepackage{graphicx}
\usepackage{cite}
\usepackage{geometry}
\newtheorem{thm}[subsection]{Theorem}
\newtheorem{lemma}[subsection]{Lemma}
\newtheorem{pro}[subsection]{Proposition}

\newtheorem{rk}[subsection]{Remark}
\newtheorem{defn}[subsection]{Definition}
\newtheorem{ex}{Example}

\numberwithin{equation}{section} 

\newcommand{\w}{{\bf w}}
\newcommand{\s}{{\sigma}}

\def \l {\lambda}

\def \s {\sigma}

\def \w {\omega}
\def \b {\beta}

\newcommand{\bea}{\begin{eqnarray}}
\newcommand{\eea}{\end{eqnarray}}

\begin{document}
\title[four Competing interactions for models]{four Competing interactions
for models with uncountable set of spin values on a Cayley Tree}

\author{U.A.Rozikov, F. H. Haydarov}

\address{U.\ A.\ Rozikov\\Institute of Mathematics and Information Technologies,\\
Tashkent, Uzbekistan,}
 \email {\tt rozikovu@yandex.ru}

\address{F.\ H.\ Haydarov\\ National University of Uzbekistan,
Tashkent, Uzbekistan.} \email {haydarov\_imc@mail.ru}

\begin{abstract} In this paper we consider four competing
interactions (external field, nearest neighbor, second neighbors
and triples of neighbors) of models with uncountable (i.e.
$[0,1]$) set of spin values on the Cayley tree of order two. We
reduce the problem of describing the "splitting Gibbs measures" of
the model to the analysis of solutions to some nonlinear integral
equation and study some particular cases for Ising and Potts
models. Also we show that periodic Gibbs measures for given models
are either translation-invariant or periodic with period two and
we give examples of the non-uniqueness of translation-invariant
Gibbs measures.
\end{abstract}
\maketitle



\section{Introduction} \label{sec:intro}

Spin models on a graph or in a continuous spaces form a large
class of systems considered in statistical mechanics. Some of them
have a real physical meaning, others have been proposed as
suitably simplified models of more complicated systems. The
geometric structure of the graph or a physical space plays an
important role in such investigations. For example, in order to
study the phase transition problem on a cubic lattice $Z^d$ or in
space $\mathbb{d}$ one uses, essentially, the Pirogov-Sinai
theory; see\cite{PS1}, \cite{PS2}\cite{12}. A general methodology
of phase transitions in $\mathbb{Z}^d$ or $\mathbb{R}^d$ was
developed in \cite{7}; some recent results in this direction have
been established in \cite{MSS}, \cite{MSSZ} (see also the
bibliography therein).

On the other hand, on a Cayley tree $\Gamma_{k}$ one uses the
theory of Markov splitting random fields based upon the
corresponding recurrent equations. In particular, in Refs
\cite{1}-\cite{3},\cite{p1a}, \cite{11}-\cite{p2},
\cite{13}-\cite{14}, \cite{16} Gibbs measures on $\Gamma_{k}$ have
been described in terms of solutions to the recurrent equations.

A number of works have been focused on various versions of the
Ising model on $\Gamma_{k}$. For example, the case
$J_{3}=\alpha=0$ was considered in \cite{gr8}, \cite{gr15} and
\cite{gr16}, where exact solutions were given, for a model with
competing restricted interactions and zero external field.(Here
and below we refer to the structure of the Hamiltonian
(\ref{e1}).) The case $J=\alpha=0$ was considered in
\cite{gr7},\cite{gr16}. In particular, Ref. \cite{gr17} proves
that there are two translation-invariant and uncountably many
non-translation-invariant extreme Gibbs measures. In \cite{gr9}
the phase transition problem was solved for $\alpha=0,\ J\cdot
J_{1}\cdot J_{3}\neq 0$ and for $J_{3}=0,\ \alpha\cdot J\cdot
J_{1}\neq 0$. In \cite{p1a} one considered Ising model with four
competing interactions (i.e., $J\cdot J_{1}\cdot J_{3}\cdot
\alpha\neq 0$ ) on $\Gamma_2$, a Cayley tree of order two. These
papers are devoted to models with a \emph{finite} set of spin
values.

In Ref. \cite{6} a Potts model with a \emph{countable} set of spin
values on a Cayley tree has been considered: it was showed that
the set of translation-invariant splitting Gibbs measures contains
at most one point, independently on parameters of the the model.
This is a crucial difference with models with finitely many spin
values: the letter may have more than one translation-invariant
Gibbs measure.

During the past five years, an increasing attention was given to
models with a \emph{uncountable} many spin values on a Cayley
tree. Until now, one considered nearest-neighbor interactions
$(J_{3}=J=\alpha=0,\ J_{1}\neq 0)$ with the set of spin values
$[0,1]$. The following results was achieved: splitting Gibbs
measures on a Cayley tree of order $k$ are described by solutions
to a nonlinear integral equation. For $k = 1$ (when the Cayley
tree becomes a one-dimensional lattice $\mathbb{Z}^1$) it has been
shown that the integral equation has a unique solution, implying
that there is a unique Gibbs measure. (Confirming a sereies of
well-known results; see, e.g., \cite{DS} and references therein.)
For a general $k$, a sufficient condition was found under which a
periodic splitting Gibbs measure is unique. On the other hand, on
a Cayley tree $\Gamma_{k}$ of order $k = 2$, phase transitions
were proven to exist. See \cite{ehr2013}-\cite{eh2015},
\cite{new1}, \cite{re}-\cite{rh2015}. We note that all of these
papers were considered for the case $J_{3}=J=\alpha=0,\ J_{1}\neq
0$.

  In this paper we describe splitting Gibbs
measures on $\Gamma_{2}$ by solutions to a nonlinear integral
equation for the case $J_{3}^{2}+ J_{1}^{2}+J^{2}+\alpha^{2}\neq
0$ which a generalization of the case $J_{3}=J=\alpha=0,\
J_{1}\neq 0$. Also we prove that periodic Gibbs measure for
Hamiltonian (\ref{e1}) with four competing interactions is either
$translation$-$invariant$ or $G_{k}^{(2)}- periodic.$ In the last
section we give examples of non-uniqueness for Hamiltonian
(\ref{e1}) in the case $J_{3}\neq 0, J=J_{1}=\alpha=0$.

\section{Preliminaries}

{\it Cayley tree.} A Cayley tree $\Gamma_k=(V,L)$ of order $k\in
\mathbb{N}$ is an infinite homogeneous tree, i.e., a graph without
cycles, with exactly $k+1$ edges incident to each vertices. Here
$V$ is the set of vertices and $L$ that of edges (arcs). Two
vertices $x$ and $y$ are called nearest neighbors if there exists
an edge $l\in L$ connecting them. We will use the notation
$l=\langle x,y\rangle$. The distance $d(x,y), x,y \in V$, on the
Cayley tree is defined by the formula

$$d(x,y)=\min\{d |\ x=x_{0},x_{1},...,x_{d-1},x_{d}=y\in V \ \emph{such that the pairs}$$
$$\langle x_{0},x_{1}\rangle,...,\langle x_{d-1},x_{d}\rangle \emph{are neighboring vertices}\}.$$\vskip
0.3 truecm

Let $x^{0}\in V$ be fixed and set

$$W_{n}=\{x\in V\ |\ d(x,x^{0})=n\}, \,\,\,\,\ V_{n}=\{x\in V\ |\ d(x,x^{0})\leq n\},$$

$$L_{n}=\{l=\langle x,y\rangle\in L\ |\ x,y \in V_{n}\}.$$\vskip
0.3 truecm
 The set of the direct successors of $x$ is denoted by $S(x),$
 i.e.
 $$S(x)=\{y\in W_{n+1}|\ d(x,y)=1\}, \ x\in W_{n}.$$
 We observe that for any vertex $x\neq x^{0},\ x$ has $k$ direct
 successors and $x^{0}$ has $k+1$. Vertices $x$ and $y$ are called second neighbors, which fact is marked as
 $\rangle x,y\langle,$ if there exist a vertex $z\in V$ such that
 $x$, $z$ and $y$, $z$ are nearest neighbors. We will consider only second neighbors $\rangle x, y \langle,$ for which there
exist $n$ such that $x, y \in W_n$. Three vertices $x,\ y$ and $z$
are called a triple of neighbors in which case we write $\langle
x, y, z\rangle,$ if $\langle x, y \rangle,\ \langle y, z \rangle$
are nearest neighbors and $x,\ z \in
W_n,\ y \in W_{n-1}$, for some $n \in \mathbb{N}$.\\

{\it Gibbs measure for models with four competing interactions.}
 We consider models with four competing interactions where the spin
 takes values in the
 unit interval $[0,1]$. Given a set $\Lambda\subset V$ a
 configuration on $\Lambda$ is an arbitrary function
$\s_\Lambda:\Lambda\to [0,1]$, with values $\s(x),\ x\in \Lambda$.
The set of all configurations on $\Lambda$ is denoted by
$\Omega_\Lambda=[0,1]^\Lambda=\Omega$ and denote by ${\mathcal B}$
the sigma-algebra generated by measurable cylinder subsets of
$\Omega$.

Fix bounded, measurable functions $\xi_{1}:(t,u,v)\in[0,1]^{3}\to
\xi_{1}(t,u,v)\in R$ and $\xi_{i}: (u,v)\in [0,1]^2\to
\xi_{i}(u,v)\in R, \ i=2,3$. We consider a model with four
competing interactions on the Cayley tree which is defined by a
formal Hamiltonian

$$H(\sigma)=-J_{3}\sum_{\langle x,y,z\rangle}\xi_{1}\left(\sigma(x),\sigma(y),\sigma(z)\right)
-J\sum_{\rangle
x,y\langle}\xi_{2}\left(\sigma(x),\sigma(z)\right)$$
\begin{equation}\label{e1}
-J_{1}\sum_{\langle
x,y\rangle}\xi_{3}\left(\sigma(x),\sigma(y)\right)-\alpha\sum_{x}\sigma(x),
\end{equation}\\
 where the sum in the first term ranges all triples of
neighbors, the second sum ranges all second neighbors, the third
sum ranges all nearest neighbors, and  $J, J_{1}, J_{3},\alpha\in
R\setminus \{0\}$.

  Hamiltonian $H(\s)$ from Eqn (2.1) generates conditional Gibbs densities.
To make a consistent definition, let $\Lambda\subset V$ be a
finite set, of cardinality $|\Lambda|$. Denoting by $\lambda$ the
Lebesgue measure on [0,1], the set of all configurations on
$\Lambda$ is equipped with an a priori measure $\lambda_\Lambda$
introduced as the $|\Lambda|$-fold power of $\lambda$.

 Let $\Lambda\subset V$ be a finite set. We denote that $\partial(\Lambda)$ is
 the set of boundary points of $\Lambda$ i.e.,
$\partial(\Lambda)=\{y\in V\setminus\Lambda\ |\ x\in \Lambda,
<x,y>\}$. Next, put
$$\Omega_{\Lambda}^{\ast}=\underbrace{\Omega_{\Lambda}\times \Omega_{\Lambda}
\times...\times \Omega_{\Lambda}}_{|\partial(\Lambda)|}\ , \ \ \
\lambda_{\Lambda}^{\ast}=\underbrace{\lambda_{\Lambda}\times
\lambda_{\Lambda} \times...\times
\lambda_{\Lambda}}_{|\partial(\Lambda)|},$$ where $\times$ is a
direct product. Let $\bar{\sigma}(V\setminus\Lambda)$ be a fixed
boundary configuration. The total energy of configuration
$\sigma=\sigma_{\Lambda}\in \Omega_{\Lambda}$ under outer
condition $\bar{\sigma}_{V\setminus\Lambda}$ is defined as

$$H(\sigma_{\Lambda}\mid \bar{\sigma}_{V\setminus\Lambda})=-J_{3}\sum_{\langle x,y,z\rangle: \ x,y,z\in \Lambda}\xi_{1}
\left(\sigma(x),\sigma(y),\sigma(z)\right) -J\sum_{\rangle
x,y\langle: \ x,y\in
\Lambda}\xi_{2}\left(\sigma(x),\sigma(y)\right)$$\\
$$-J_{1}\sum_{\langle x,y\rangle: \ x,y\in \Lambda}\xi_{3}\left(\sigma(x),\sigma(y)\right)-\alpha\sum_{x\in
\Lambda}\sigma(x)-J_{3}\sum_{\langle x,y,z\rangle: \ x\in \Lambda
\textrm{and}\
z\notin \Lambda}\xi_{1}\left(\sigma(x),\sigma(y),\sigma(z)\right)$$\\
\begin{equation}\label{addition}-J\sum_{\rangle x,y\langle: \ x\in \Lambda, \ y\notin \Lambda
}\xi_{2}\left(\sigma(x),\bar{\sigma}(y)\right) -J_{1}\sum_{\langle
x,y\rangle: \ x\in \Lambda, \ y\notin \Lambda}
\xi_{2}\left(\sigma(x),\bar{\sigma}(y)\right),\end{equation} where
the first and forth sums are taken over triple of neighbors; the
second and sixth sums are taken over second neighbors and the
third and fifth sums are taken over nearest neighbors.

 For a configuration $\sigma_{\Lambda}: \Lambda\rightarrow [0,1]$ the
conditional Gibbs density is defined as

$$ \nu^{\Lambda}_{\bar{\sigma}_{V\setminus\Lambda}}(\sigma_{\Lambda})=\frac{1}{Z_{\Lambda}\left(
\bar{\sigma}_{V\setminus\Lambda}\right)} \exp \left(-\beta H
\left(\sigma_{\Lambda}\,||\,
\bar{\sigma}_{V\setminus\Lambda}\right)\right),$$ where
$\beta=\frac{1}{T},\ T>0,$ and
$Z_{\Lambda}\left(\bar{\sigma}_{V\setminus\Lambda}\right)$ is a
partition function, i.e.,

 \vskip 0.3truecm
$$Z_{\Lambda}\left(\bar{\sigma}_{V\setminus\Lambda})\right)=\int\!\!\!...\!\!\!\!
\int\limits_{\Omega_{\Lambda}^{\ast}\ \ }\exp \left(-\beta H
\left(\sigma_{\Lambda}\,||\,
\bar{\sigma}|_{\partial(\Lambda)}\right)\right)(\lambda^{\ast}_
{\Lambda})(d\sigma_{\Lambda}).$$

We note that if $x\in \Lambda$, $y\in V\setminus\Lambda$ and
$\langle x,y\rangle$ then $y \in \partial(\Lambda)$. Therefore, we
can exchange $\bar{\sigma}_{V\setminus\Lambda}$ for
$\partial(\Lambda).$ Finally, the conditional Gibbs measure
$\mu_{\Lambda}$ in volume $\Lambda$ under the boundary condition
$\bar{\sigma}|_{\partial(\Lambda)}$ is defined by

\begin{equation}\label{e2} \mu\left(\sigma\in \Omega:\
\sigma|_{\Lambda}=\sigma_{\Lambda}\right)=\int\!\!\!...\!\!\!\!
\int\limits_{\Omega_{\Lambda}^{\ast}\ \
}(\lambda^{\ast}_{\Lambda})(d\sigma_{\Lambda})
\nu^{\Lambda}_{\bar{\sigma}|_{\partial(\Lambda)}}(\sigma_{\Lambda}).
\end{equation}

\section{The integral equation}

 Let $h: [0,1]\times V\setminus \{x^{0}\}\rightarrow \mathbb{R}$ and
  $|h(t,x)|=|h_{t,x}|<C$ where $x_{0}$ is a root of Cayley tree and $C$ is a
constant which does not depend on $t$.  For some $n\in\mathbb{N}$
and $\sigma_n:x\in V_n\mapsto \sigma(x)$ we consider the
probability distribution $\mu^{(n)}$ on $\Omega_{V_n}$ defined by
\begin{equation}\label{e2}\mu^{(n)}(\sigma_n)=Z_n^{-1}\exp\left(-\beta H(\sigma_n)
+\sum_{x\in W_n}h_{\sigma(x),x}\right),\end{equation}\\
 where
$Z_n$ is the corresponding partition function:
\begin{equation}\label{e3}Z_n=\int\!\!\!...\!\!\!\!\!\int\limits_{\Omega^{\ast}_{V_{n-1}}} \exp\left(-\beta
H({\widetilde\sigma}_n) +\sum_{x\in
W_{n}}h_{{\widetilde\sigma}(x),x}\right)
\lambda^{\ast}_{V_{n-1}}({d\widetilde\s_n}),\end{equation}

Let $\sigma_{n-1}\in\Omega_{V_{n-1}}$ and
$\sigma_{n-1}\vee\omega_n\in\Omega_{V_n}$ is the concatenation of
$\sigma_{n-1}$ and $\omega_n.$ For $n\in \mathbb{N}$ we say that
the probability distributions $\mu^{(n)}$ are compatible if
$\mu^{(n)}$ satisfies the following condition:\vskip 0.1 truecm
\begin{equation}\label{e4}\int\!\!\!\!\!\!\!\!\!\int\limits_{\Omega_{W_n}\times\Omega_{W_n}}
\mu^{(n)}(\sigma_{n-1}\vee\omega_n)(\lambda_{W_n}\times
\lambda_{W_n})(d\omega_n)=
\mu^{(n-1)}(\sigma_{n-1}).\end{equation}\vskip 0.1 truecm

By Kolmogorov's extension theorem there exists a unique measure
$\mu$ on $\Omega_V$ such that, for any $n$ and
$\sigma_n\in\Omega_{V_n}$, $\mu \left(\left\{\sigma
|_{V_n}=\sigma_n\right\}\right)=\mu^{(n)}(\sigma_n)$. The measure
$\mu$ is called {\it splitting Gibbs measure} corresponding to
Hamiltonian (\ref{e1}) and function $x\mapsto h_x$, $x\neq x^0$.\\
Denote
\begin{equation}\label{e20}K(t,u,v)=\exp\left\{J_{3}\beta\xi_{1}\left(t,u,v\right)+J\beta\xi_{2}\left(u,v\right)
+J_{1}\beta\left(\xi_{3}\left(t,u\right)+\xi_{3}\left(t,v\right)\right)+\alpha\beta(u+v)\right\},\end{equation}

$$\underbrace{\Omega_{W_{n}}\times\Omega_{W_{n}}\times...\times\Omega_{W_{n}}}_{3\cdot
2^{p-1}}=\Omega^{(p)}_{W_{n}},\ \ \
\underbrace{\lambda_{W_{n}}\times\lambda_{W_{n}}\times...\times\lambda_{W_{n}}}_{3\cdot
2^{p-1}}=\lambda^{(p)}_{W_{n}}, \ n,p\in \mathbb{N},$$ and
$$f(t,x)=\exp(h_{t,x}-h_{0,x}), \ \ (t,u,v)\in [0,1]^{3},\ x\in
V\setminus\{x^{0}\}.$$\vskip 0.3truecm
\begin{lemma}\label{l1} Let $\w_{n}(\cdot): W_{n}\rightarrow
[0,1], \ n\geq 2$. Then the following equality holds:

$$\int\!...\!\!\!\!\!\int\limits_{\Omega^{(n)}_{W_n}}
\prod_{x\in W_{n-1}}\prod_{\rangle y,z\langle\in
S(x)}K\left(\w_{n-1}(x),\w_{n}(y),\w_{n}(z)\right)f(\w_{n}(y),y)f(\w_{n}(z),z)d(\w_{n}(y))d(\w_{n}(z))=$$
$$\prod_{x\in W_{n-1}}\prod_{\rangle y,z\langle\in
S(x)} \ \int\!\!\!\!
\int\limits_{\Omega_{W_n}^{(2)}}K\left(\w_{n-1}(x),\w_{n}(y),\w_{n}(z)\right)f(\w_{n}(y),y)f(\w_{n}(z),z)d(\w_{n}(y))d(\w_{n}(z)).$$
\end{lemma}

\begin{proof} Denote elements of $W_{n-1}$ by $x_{i},$
 i.e.,
 $$x_{i}\in W_{n-1}, \
i\in\{1,2,...,3\cdot2^{n-2}\} , \bigcup_{i=1}^{\ \ \ \
3\cdot2^{n-2}}\{x_{i}\}=W_{n-1}\ \textrm{and} \
S(x_{i})=\{y_{i},z_{i}\}.$$ Then
$$\int\!...\!\!\!\!\!\int\limits_{\Omega^{(n)}_{W_n}}
\prod_{x\in W_{n-1}}\prod_{\rangle y,z\langle\in
S(x)}K\left(\w_{n-1}(x),\w_{n}(y),\w_{n}(z)\right)f(\w_{n}(y),y)f(\w_{n}(z),z)d(\w_{n}(y))d(\w_{n}(z))=$$
\begin{equation}\label{e7}\int\!...\!\!\!\!\!\int\limits_{\Omega^{(n)}_{W_n}}\prod_{i=1}^{\
\ \ 3\cdot2^{n-2}}
K\left(\w_{n-1}(x_{i}),\w_{n}(y_{i}),\w_{n}(z_{i})\right)f(\w_{n}(y_{i}),y_{i})f(\w_{n}(z_{i}),z_{i})d(\w_{n}(y_{i}))d(\w_{n}(z_{i})).\end{equation}\\
Since $\w_{n}(y_{i}), i\in\{1,2,...,3\cdot2^{n-2}\}$ and
$\w_{n}(z_{j}), j\in\{1,2,...,3\cdot2^{n-2}\}$ are independent
configurations, the RHS of (\ref{e7}) is equal to
$$\zeta(\w_{n-1}(x_{1}),y_{1},z_{1})\int\!...\!\!\!\!\!\int\limits_{\Omega^{(n-2)}_{W_n}}K\left(\w_{n-1}(x_{2}),\w_{n}
(y_{2}),\w_{n}(z_{2})\right)...K\left(\w_{n-1}(x_{3\cdot2^{n-2}}),
\w_{n}(y_{3\cdot2^{n-2}}),\w_{n}(z_{3\cdot2^{n-2}})\right)$$
$$\times f(\w_{n}(y_{2}),y_{2})f(\w_{n}(z_{2}),z_{2})...f(\w_{n}(y_{3\cdot2^{n-2}}),y_{3\cdot2^{n-2}})
f(\w_{n}(z_{3\cdot2^{n-2}}),z_{3\cdot2^{n-2}})d(\w_{n}(y_{2}))d(\w_{n}(z_{2}))...$$\vskip
0.1 truecm
\begin{equation}\label{e8}...d(\w_{n}(y_{3\cdot2^{n-2}}))d(\w_{n}(z_{3\cdot2^{n-2}})),\end{equation}
where

$$\zeta(\w_{n-1}(x_{i}),y_{i},z_{i})=
\int\!\!\int\limits_{\Omega_{W_{n}}^{(2)}}K\left(\w_{n-1}(x_{i}),\w_{n}(y_{i}),\w_{n}(z_{i})\right)f(\w_{n}(y_{i}),y_{i})
f(\w_{n}(z_{i}),z_{i})d(\w_{n}(y_{i}))d(\w_{n}(z_{i})).$$\\
Continuing this process, equation (\ref{e8}) can be written as

$$\prod_{i=1}^{\ \ \ 3\cdot2^{n-2}}
\zeta(\w_{n-1}(x_{i}),y_{i},z_{i})=$$
$$\prod_{i=1}^{\ \ \ 3\cdot2^{n-2}}\int\!\!\int\limits_
{\Omega_{W_{n}}^{(2)}}K\left(\w_{n-1}(x_{i}),\w_{n}(y_{i}),\w_{n}(z_{i})\right)f(\w_{n}(y_{i}),y_{i})
f(\w_{n}(z_{i}),z_{i})d(\w_{n}(y_{i}))d(\w_{n}(z_{i}))=$$
$$\prod_{x\in W_{n-1}}\prod_{>y,z<\in
S(x)} \ \int\!\!\!
\int\limits_{\Omega_{W_n}^{(2)}}K\left(\w_{n-1}(x),\w_{n}(y),\w_{n}(z)\right)f(\w_{n}(y),y)f(\w_{n}(z),z)d(\w_{n}(y))d(\w_{n}(z)).$$\\
This completes the proof.\end{proof}

The following statement describes conditions on $h_x$ guaranteeing
compatibility of the corresponding distributions
$\mu^{(n)}(\sigma_n).$

 \begin{thm} \label{th1}  The measure
$\mu^{(n)}(\sigma_n)$, $n=1,2,\ldots$ satisfies the consistency
condition (\ref{e4}) iff for any $x\in V\setminus\{x^0\}$ the
following equation holds:
\begin{equation}\label{e5} f(t,x)=\prod_{\rangle y,z\langle\in S(x)}
\frac{\int_0^1\int_0^1K(t,u,v)f(u,y)f(v,z)dudv}{\int_0^1\int_0^1K(0,u,v)f(u,y)f(v,z)dudv},
\end{equation} here $S(x)=\{y,z\},\ \langle y,x,z\rangle$
 is a ternary neighbor and $du=\l(du)$ is the Lebesgue measure. \end{thm}

\begin{proof} {\sl Necessity.}  Suppose that (\ref{e4}) holds; we want to
prove (\ref{e5}). Substituting (\ref{e2}) in (\ref{e4}) we obtain
that for any configurations $\sigma_{n-1}$: $x\in
V_{n-1}\mapsto\sigma_{n-1}(x)\in [0,1] $:

$$\frac{Z_{n-1}}{Z_n}\int\!...\!\!\!\!\!\int\limits_{\Omega^{(n)}_{W_n}}
\exp\left(J_{3}\beta\sum_{\langle y,x,z\rangle,x\in
W_{n-1}}\xi_{1}\left(\sigma_{n-1}(x),\sigma_{n}(y),\sigma_{n}(z)\right)\right)\times$$

$$\exp\left(J\beta\sum_{\rangle y,z\langle\in
W_{n}}\xi_{2}(\sigma_{n}(y),\sigma_{n}(z))+J_{1}\beta\sum_{\langle
x,y\rangle, x\in
W_{n-1}}\xi_{3}\left(\sigma_{n-1}(x),\sigma_{n}(y)\right)\right)\times$$\vskip
0.3 truecm
$$\exp\left(\alpha\beta\sum_{y\in S(x), x\in
W_{n-1}}\sigma_{n}(y)+\sum_{y\in S(x), x\in
W_{n-1}}h_{\w_{n}(y),y}\right)\l^{(n)}_{W_n}(d\w_n)=\exp\left(\sum_{x\in
W_{n-1}}h_{\sigma_{n-1}(x),x}\right),$$\\
where $\omega_n$: $x\in W_n\mapsto\omega_n(x)$. From the last
equality we get:

$$\frac{Z_{n-1}}{Z_n}\int\!...\!\!\!\!\!\int\limits_{\Omega^{(n)}_{W_n}}
\prod_{x\in W_{n-1}}\prod_{\rangle y,z\langle\in
S(x)}\exp\left(J_{3}\beta\sum_{\langle
y,x,z\rangle}\xi_{1}\left(\sigma_{n-1}(x),\w_{n}(y),\w_{n}(z)\right)\right)\times$$
$$\exp\left(J\beta\sum_{\rangle y,z\langle}\xi_{2}(\w_{n}(y),\w_{n}(z))+
J_{1}\beta\cdot\xi_{3}(\sigma_{n-1}(x),\w_{n}(y))+J_{1}\beta\cdot\xi(\sigma_{n-1}(x),\w_{n}(z))\right)\times$$
$$\exp\left(\alpha\beta(\w_{n}(y)+\w_{n}(z))+h_{\w_{n}(y),y}+h_{\w_{n}(z),z}\right)d(\w_{n}(y))d(\w_{n}(z))=\exp\left(\sum_{x\in
W_{n-1}}h_{\sigma_{n-1}(x),x}\right).$$\\
By Lemma \ref{l1}
$$\frac{Z_{n-1}}{Z_n}\prod_{x\in W_{n-1}}\prod_{\rangle y,z\langle\in
S(x)}\int\!\!\!\int\limits_{\Omega^{(2)}_{W_n}}
\exp\left(J_{3}\beta\sum_{<y,x,z>}\xi_{1}\left(\sigma_{n-1}(x),\w_{n}(y),\w_{n}(z)\right)\right)\times$$
$$\exp\left(J\beta\sum_{\rangle y,z\langle}\xi_{2}(\w_{n}(y),\w_{n}(z))+
J_{1}\beta\cdot\xi_{3}(\sigma_{n-1}(x),\w_{n}(y))+J_{1}\beta\cdot\xi(\sigma_{n-1}(x),\w_{n}(z))\right)\times$$
$$\exp\left(\alpha\beta(\w_{n}(y)+\w_{n}(z))+h_{\w_{n}(y),y}+h_{\w_{n}(z),z}\right)d(\w_{n}(y))d(\w_{n}(z))=\exp\left(\sum_{x\in
W_{n-1}}h_{\sigma_{n-1}(x),x}\right).$$\\
 Consequently, for any $\sigma_{n-1}(x)\in [0,1]$, $f(\sigma_{n-1}(x),x)$ is equal to\\
$$\prod_{\rangle y,z\langle\in S(x)}
\frac{\int\!\!\int_{\Omega_{W_{n}}^{(2)}}K(\sigma_{n-1}(x),\w_{n}(y),\w_{n}(z))f(\w_{n}(y),y)f(\w_{n}(z),z)d(\w_{n}(y))d(\w_{n}(z))}
{\int\!\!\int_{\Omega_{W_{n}}^{(2)}}K(0,\w_{n}(y),\w_{n}(z))f(\w_{n}(y),y)f(\w_{n}(z),z)d(\w_{n}(y))d(\w_{n}(z))}.$$\\
 If we denote $\w_{n}(y)=u,\ \w_{n}(z)=v,\ \sigma_{n-1}(x)=t$ it will imply (\ref{e5}).\vskip 0.3truecm

{\sl Sufficiency.} Suppose that (\ref{e5}) holds. It is equivalent
to the representations
\begin{equation}\label{e9}\prod_{\rangle y,z\langle\in
S(x)}\int\!\!\int\limits_{\Omega_{W_{n}}^{(2)}}
K(t,u,v)\exp(h_{u,y}+h_{v,z})dudv= a(x)\exp\,(h_{t,x}), \ t\in
[0,1] \end{equation}
 for some function
$a(x)>0, x\in V.$ We have
$${\rm LHS \ \ of \ \  (\ref{e5})}=\frac{1}{Z_n}\exp(-\b H(\s_{n-1}))
\l^{\ast}_{V_{n-2}}(d(\s_{n-1}))\times$$
$$\prod_{x\in W_{n-1}}
\prod_{\rangle y,z\langle\in
S(x)}\int\!\!\!\int\limits_{\Omega_{W_{n}}^{(2)}}\
exp\left(J_{3}\beta\sum_{\langle
y,x,z\rangle}\xi_{1}\left(\sigma_{n-1}(x),u,v\right)+J\beta\sum_{\rangle
y,z\langle}\xi_{2}(u,v)+
J_{1}\beta\cdot\xi_{3}(\sigma_{n-1}(x),u)\right)$$
\begin{equation}\label{e10}\times\exp\left(J_{1}\beta\cdot\xi_{3}(\sigma_{n-1}(x),v)+
\alpha\beta(u+v)+h_{u,y}+h_{v,z}\right)dudv=\exp\left(\sum_{x\in
W_{n-1}}h_{\sigma_{n-1}(x),x}\right).\end{equation}\\
  Let $A_n(x)=\prod_{x \in W_{n-1}}
a(x)$, then from (\ref{e9}) and (\ref{e10}) we get
\begin{equation}\label{e11}{\rm RHS\ \  of\ \  (\ref{e10}) }=
\frac{A_{n-1}}{Z_n}\exp(-\b H(\s_{n-1}))
\lambda^{\ast}_{V_{n-2}}(d\s)\prod_{x\in W_{n-1}}
h_{\s_{n-1}(x),x}.\end{equation}\\
Since $\mu^{(n)}$, $n \in \mathbb{N}$ is a probability
distribution, we should have

$$\int\!\!...\!\!\!\!\!\!\int\limits_{\Omega^{\ast}_{V_{n-2}}}\lambda^{\ast}_{V_{n-2}}(d\s_{n-1})
\int\!\!\!\int\limits_{\Omega_{W_{n}}^{(2)}}
\lambda^{(2)}_{W_{n}}(d\omega_{n}) \mu^{(n)} (\sigma_{n-1},
\omega_n) = 1 .$$
 Hence from (\ref{e11})  we get $Z_{n-1}A_{n-1}=Z_n$, and
(\ref{e5}) holds. Theorem is proved. \end{proof}

Note that in all of papers \cite{ehr2013}-\cite{eh2015},
\cite{new1}, \cite{re}-\cite{rh2015} were considered the
Hamiltonian (\ref{e1}) for the case $J_{3}=J=\alpha=0$ and
$J_{1}\neq 0$ and it was proved that: The probability
distributions $\mu^{(n)}(\sigma_n)$, $n=1,2,\ldots$ are compatible
iff for any $x\in V\setminus\{x^0\}$ the following equation holds:

\begin{equation}\label{e17}
f(t,x)=\prod_{y\in
S(x)}\frac{\int_0^1\exp\left\{J_{1}\beta\xi_{3}(t,u)\right\}
f(u,y)du}{\int_0^1\exp\left\{J_{1}\beta\xi_{3}(0,u)\right\}
f(u,y)du},\end{equation}
 where $f(t,x)=\exp(h_{t,x}-h_{0,x}), \
t\in [0,1],\ x\in V.$

 Equation (\ref{e17}) was first considered in \cite{re}. The
following remark gives us equation (\ref{e5}) is coincide with
equation (\ref{e17}) in the case $J_{3}=J=\alpha=0, J_{1}\neq 0$.

\begin{rk}\label{rem} If $J_{3}=J=\alpha=0$ and
$J_{1}\neq 0$ then (\ref{e5}) is equivalent to (\ref{e17}).
\end{rk}
\begin{proof} For $J_{3}=J=\alpha=0$ and $J_{1}\neq 0$ one get
$K(t,u,v)=\exp\left\{J_{1}\beta\left(\xi_{3}\left(u,t\right)+\xi_{3}\left(v,t\right)\right)\right\}.$
Then (\ref{e5}) can be written as $$f(t,x)=\prod_{\rangle
y,z\langle\in S(x)}
\frac{\int_0^1\int_0^1\exp\left\{J_{1}\beta\left(\xi_{3}\left(t,u\right)+\xi_{3}\left(t,v\right)\right)\right\}
f(u,y)f(v,z)dudv}{\int_0^1\int_0^1\exp\left\{J_{1}\beta\left(\xi_{3}\left(0,u\right)+\xi_{3}\left(0,v\right)\right)\right\}f(u,y)f(v,z)dudv}=$$\vskip0.3truecm

\begin{equation}\label{e16}\prod_{\rangle y,z\langle\in
S(x)}\frac{\int_0^1\exp\left\{J_{1}\beta\xi_{3}(t,u)\right\}
f(u,y)du \cdot \int_0^1\exp\left\{J_{1}\beta\xi_{3}(t,v)\right\}
f(v,z)dv}{\int_0^1\exp\left\{J_{1}\beta\xi_{3}(0,u)\right\}
f(u,y)du \cdot \int_0^1\exp\left\{J_{1}\beta\xi_{3}(0,v)\right\}
f(v,z)dv}.\end{equation}\\
 Since $\rangle y,z\langle=S(x)$ equation (\ref{e16}) is equivalent to
 (\ref{e17}).
\end{proof}

{\it \textbf{The Ising model with competing interactions}.} \ It's
known that if $\xi_{1}(x,y,z)=xyz,\ \xi_{i}(x,y)=xy,\ i\in\{2,3\}$
then model (\ref{e1}) become the Ising model with uncountable set
of spin values. For the case $J_{1}=J_{3}=0$ and $J\neq 0,\
\alpha\in\mathbb{R}$ it's clear that (\ref{e5}) is equivalent
to\\
$$f(t,x)=\prod_{\rangle y,z\langle\in S(x)}
\frac{\int_0^1\int_0^1\exp\{J\beta
uv+\alpha\beta(u+v)\}f(u,y)f(v,z)dudv}{\int_0^1\int_0^1\exp\{J\beta
uv+\alpha\beta(u+v)\}f(u,y)f(v,z)dudv}=1.$$\\
As a result, equation (\ref{e5}) has the unique solution
$f(t,x)=1,\ t\in [0,1],\ x\in V$ for any $\beta>0$. Consequently
we get following Proposition.

\begin{pro}\label{p2} Let $J_{1}=J_{3}=0$ and $J\neq 0,\ \alpha\in\mathbb{R}$.
Then the Ising model with uncountable set of spin values on Cayley
tree of order two has unique splitting Gibbs measures for any
$J\in \mathbb{R},$ and any $\beta>0$.\end{pro}

{\it \textbf{The Potts Model with competing interactions}.} \ Put
$J_{3}=0$ and $J, J_{1}, \alpha \in \mathbb{R}$. If
$\xi_{i}(x,y)=\delta(x,y),\ i\in\{2,3\}$ ($\delta$ is the
Kronecker's symbol) then the model (\ref{e1}) become Potts model.
For any $t\in [0,1],\ x\in V$ it's easy to see that
$$\int_{0}^{1}\!\!\int_{0}^{1}\exp\{J\beta\delta(u,v)
+J_{1}\beta(\delta(u,t)+\delta(v,t))+\alpha\beta(u+v)\}dudv=\ \ \
\ \  \ \ \ \ \ \ \ \ \  \ \ \ \ \ \ \ \ \  $$
$$  \ \ \ \ \  \ \ \ \ \ \ \ \ \ \ \ \ \ \ \ \ \ \int_{0}^{1}\!\!\int_{0}^{1}\exp\{J\beta\delta(u,v)
+J_{1}\beta(\delta(u,0)+\delta(v,0))+\alpha\beta(u+v)\}dudv.$$\\
Hence in this case the equation has the unique solution $f(t,x)=1$
and we can conclude that

\begin{pro}\label{p3} The Potts model with uncountable set of spin
values on Cayley tree of order two has unique splitting Gibbs
measure for any $J_{3}\neq 0$ and $J, J_{1}, \alpha \in
\mathbb{R}, \ \beta>0.$
\end{pro}

\begin{rk} For $J_{3}\cdot J_{1}\cdot J\cdot\alpha\neq 0$ is there a kernel
$K(t,u,v)>0$ of equation (\ref{e5}) when the equation has at least
two solutions? This is an open problem.
\end{rk}
\section{Periodic Gibbs measure of the model (\ref{e1})}

In this section we consider periodic Gibbs measures of the model
(\ref{e1}) and give a result (Theorem \ref{thm2}) about periodic
Gibbs measures for the model.

Let $G_{k}$ be a free product of $k+1$ cyclic groups of the second
order with generators $a_{1},a_{2},...a_{k+1},$ respectively.
There exist bijective maps from the set of vertices $V$ of the
Cayley tree $\Gamma_{k}$ onto the group $G_{k}$ (see \cite{uar}).
That's why we sometimes replace $V$ with $G_{k}$.

 Let $S_{1}(x)=\{y\in G_{k}:  \langle x,y\rangle\}$ be the collection of all
neighbors to the word $x\in G_{k}.$ Let $G^\ast$ be a normal
subgroup of index $r$ in $G_{k},$ and let
$G_{k}/G^\ast=\{G^\ast_{0}, G^\ast_{1}, ..., G^\ast_{r-1}\}$ be a
quotient group, with the coset $G^\ast_{0}=G^\ast.$ In addition,
let $q_{i}(x)=|S_{1}(x)\bigcap G^\ast_{i}|, i=0,1,...,r-1,$ and
$Q(x)=(q_{0}(x), q_{1}(x),...,q_{r-1}(x))$ where $x\in G_{k},$ \,\
$q_{i}(G^\ast_{0})=q_{i}(e)=|\{j : a_{j}\in G^\ast_{i}\}|, \,\
Q(G^\ast_{0})=(q_{0}(G^\ast_{0}),...,q_{n-1}(G^\ast_{0})).$

\begin{defn}\label{d1} Let $G^\ast$ be a subgroup of $G_{k}, k\geq 1$. We say that a
function $h_{x}, x\in G_{k}$ is $K$-periodic if $h_{yx}=h_{x}$ for
all $x\in G_{k}$, $y\in K$. A $G_{k}$- periodic function $h$ is
called $translation$-$invariant.$\end{defn}

\begin{defn}\label{d1} A Gibbs measure is called $G^\ast$- periodic if it corresponds to a $G^\ast$-
periodic function $h$.
\end{defn}

\begin{pro}\label{p1} \cite{uar} For any $x\in G_{k},$ there exists a permutation
$\pi_{x}$ of the coordinates of the vector $Q(G^\ast_{0})$
such that $\pi_{x}(Q(G^\ast_{0}))=Q(x).$\\
Let $G_{k}^{(2)}=\{x\in G_{k}: $ the length of word $x$ is
even$\}.$\end{pro}

Put
\begin{equation}\label{e12}
\Re^+=\left\{\vartheta_{1}(z_{1},z_{2})\vartheta_{2}(z_{1},z_{3})\
|\ \vartheta_{i}\in C\left([0,1]\times[0,1]\right), \
\vartheta_{i}(\cdot,\cdot)>0, \ i\in\{1,2\} \right\}.
\end{equation}

\begin{thm}\label{thm2} Let $K(z_{1},z_{2},z_{3})\in\Re^+$ and $G^\ast$ be a normal subgroup of finite index
in $G_{k}.$ Then each $G^\ast$- periodic Gibbs measure for the
model (\ref{e1}) is either $translation$-$invariant$ or
$G_{k}^{(2)}- periodic.$\end{thm}

\begin{proof} By Theorem \ref{th1}
$$f(\sigma_{n-1}(x),x)=\prod_{\rangle y,\ z\langle\in S(x)}
\frac{\int\!\!\int_{\Omega^{(2)}_{W_n}}K\left(\sigma_{n-1}(x),\w_{n}(y),\w_{n}(z)\right)
f(\w_{n}(y),y)f(\w_{n}(z),z)d(\w_{n}(y))d(\w_{n}(z))}{\int\!\!\int_{\Omega^{(2)}_{W_n}}K\left(0,\w_{n}(y),\w_{n}(z)\right)
f(\w_{n}(y),y)f(\w_{n}(z),z)d(\w_{n}(y))d(\w_{n}(z))}$$\\
Let $\{x_\downarrow, y, z\}=S_{1}(x).$ From Proposition \ref{p1}

  $$f(\sigma_{n-1}(x),x)=\prod_{\rangle y,\ z\langle\in S(x)}\frac{\int\!\!\int_{\Omega^{(2)}_{W_n}}K\left(\sigma_{n-1}(x),\w_{n}(y),\w_{n}(z)\right)
f(\w_{n}(y),y)f(\w_{n}(z),z)d(\w_{n}(y))d(\w_{n}(z))}{\int\!\!\int_{\Omega^{(2)}_{W_n}}K\left(0,\w_{n}(y),\w_{n}(z)\right)
f(\w_{n}(y),y)f(\w_{n}(z),z)d(\w_{n}(y))d(\w_{n}(z))}.$$

$$=\prod_{\rangle y,\ x_{\downarrow}\langle\in S(x)}\frac{\int\!\!\int_{\Omega^{(2)}_{W_n}}K\left(\sigma_{n-1}(x),\w_{n}(y),\w_{n}( x_{\downarrow})\right)
f(\w_{n}(y),y)f(\w_{n}( x_{\downarrow}),
x_{\downarrow})d(\w_{n}(y))d(\w_{n}(x_{\downarrow}))}{\int\!\!\int_{\Omega^{(2)}_{W_n}}K\left(0,\w_{n}(y),\w_{n}(
x_{\downarrow})\right) f(\w_{n}(y),y)f(\w_{n}( x_{\downarrow}),
x_{\downarrow})d(\w_{n}(y))d(\w_{n}( x_{\downarrow}))}.$$\\
 From $K(z_{1},z_{2},z_{3})\in\Re^+$
there exist $K_{1}(z_{1},z_{2})$ and $K_{2}(z_{1},z_{3})$ such
that $K(z_{1},z_{2},z_{3})=K_{1}(z_{1},z_{2})K_{2}(z_{1},z_{3}).$
As a result, we get
$$\frac{\int_{\Omega_{W_n}}K_{2}\left(\sigma_{n-1}(x),\w_{n}(z)\right)
f(\w_{n}(z),z)d(\w_{n}(z))}{\int_{\Omega_{W_n}}K_{2}\left(0,\w_{n}(z)\right)
f(\w_{n}(z),z)d(\w_{n}(z))}=\ \ \ \ \ \ \ \ \ \ \ \ \ \ \ \ \ \  \
\ \ \ \ \ \ \ \ \ \ \ \ \ \ \ \ \ \ \ \ \ \ \ \ \ \ \ \ \ $$
\begin{equation}\label{e13} \ \ \ \ \ \ \ \  \ \ \  \ \ \  \ \ \ \ \ \ \  \ \ \ \  \ \ \ \ \  \
 =\frac{\int_{\Omega_{W_n}}K_{2}\left(\sigma_{n-1}(x),\w_{n}(x_{\downarrow})\right)
f(\w_{n}(x_{\downarrow}),x_{\downarrow})d(\w_{n}(x_{\downarrow}))}{\int_{\Omega_{W_n}}K_{2}\left(0,\w_{n}(x_{\downarrow})\right)
f(\w_{n}(x_{\downarrow}),x_{\downarrow})d(\w_{n}(x_{\downarrow}))}.\end{equation}\\
Let $\w_{n}(x_{\downarrow})=p,\ \w_{n}(y)=u, \ \w_{n}(z)=v$ and
$\sigma_{n-1}(x)=t.$  Then (\ref{e13}) can be written as
\begin{equation}\label{e14}\frac{\int_{0}^{1}K_{2}(t,v)h(v,z)dv}{\int_{0}^{1}K_{2}(0,v)h(v,z)dv}=
\frac{\int_{0}^{1}K_{2}(t,p)h(p,x_{\downarrow})dp}{\int_{0}^{1}K_{2}(0,p)h(p,x_{\downarrow})dp}.\end{equation}\\
Similarly, we get
\begin{equation}\label{e15}\frac{\int_{0}^{1}K_{1}(t,u)h(u,y)du}{\int_{0}^{1}K_{1}(0,u)h(u,y)du}=
\frac{\int_{0}^{1}K_{1}(t,p)h(p,x_{\downarrow})dp}{\int_{0}^{1}K_{1}(0,p)h(p,x_{\downarrow})dp}.\end{equation}\\
By (\ref{e14}) and (\ref{e15})
 $$h(t,x)=\frac{\int_0^1\!\int_{0}^{1}K(t,p_{1},p_{2})h(p_{1},x_{\downarrow})h(p_{2},x_{\downarrow})dp_{1}dp_{2}}
 {\int_0^1\!\int_{0}^{1}K(0,p_{1},p_{2})h(p_{1},
 x_{\downarrow})h(p_{2},x_{\downarrow})dp_{1}dp_{2}}.$$\\
Analogously,
$$h(\w_{n-1}(x),y)=\frac{\int_0^1\!\int_{0}^{1}K(\w_{n-1}(x),p_{1},p_{2})h(p_{1},x)h(p_{2},x)dp_{1}dp_{2}}
 {\int_0^1\!\int_{0}^{1}K(0,p_{1},p_{2})h(p_{1},x)h(p_{2},x)dp_{1}dp_{2}}
 =h(\w_{n-1}(x),z).$$\\
 From the last equation and Proposition \ref{p1} we get
$h(\cdot,y)=h(\cdot,z)=h(\cdot,x_{\downarrow})=h_{1}$ and
$h(\cdot,x)=h_{2}$. If $h_{1}=h_{2}$ then the corresponding
measure is \emph{translation-invariant} and if $h_{1}\neq h_{2}$
then it is $G_{k}^{(2)}-$ \emph{periodic}. This completes the
proof.\end{proof}

\begin{rk}\label{rem0} Theorem \ref{thm2} is a generalization of Theorem 1 in
\cite{rh2015} from the case $J_{3}=J_{1}=\alpha=0, J\neq 0$ to the
case $J_{3}^{2}+J_{1}^{2}+J^{2}+\alpha^{2}\neq 0$. Indeed, if
 $J_{3}=J=\alpha=0,\ J_{1}\neq 0$ then
$K(t,u,v)=\exp\{J\beta\xi_{3}(t,u)\}\exp\{J\beta\xi_{3}(t,v)\}=\vartheta(t,u)\cdot\vartheta(t,v)\in
\Re^+.$ \end{rk}

 Theorem \ref{thm2} reduces the problem of
finding $H$-periodic solutions of (\ref{e5}) to finding of
$G^{(2)}_k$ -periodic or \emph{translation-invariant} solutions to
(\ref{e5}). We say that function $f(t,x)$ is a
translation-invariant if, for some function $f_{1}(t),$
$f(t,x)=f(t)$, for all $x\in V$. Similarly, $f(t,x)$ is
$G^{(2)}_k$ -periodic if, for some functions $f_1(t)$ and
$f_2(t)$,
 \[ f(t, x) =
\begin{cases}
f_{1}(t) & \text{if $x\in G_{k}^{(2)}$}\,; \\
f_{2}(t) & \text{if $x\in G_{k}\setminus G_{k}^{(2)}$}\,.
\end{cases} \]\\
Consequently, for $K(\alpha,\beta,\gamma)\in\Re^+$ it remains to
study only two equations:

\begin{equation}\label{e18}f(t)=\frac{\int_0^1\!\!\int_0^1K(t,u,v)f(u)f(v)dudv}{\int_0^1\!\!\int_0^1K(0,u,v)f(u)f(v)dudv},
\end{equation}
and
\begin{equation}\label{e19}f(t)=\frac{\int_0^1\!\!\int_0^1K(t,u,v)g(u)g(v)dudv}{\int_0^1\!\!\int_0^1K(0,u,v)g(u)g(v)dudv},
\ \ \
g(t)=\frac{\int_0^1\!\!\int_0^1K(t,u,v)f(u)f(v)dudv}{\int_0^1\!\!\int_0^1K(0,u,v)f(u)f(v)dudv}.
\end{equation}
\begin{ex}\label{ex1} If $K(t,u,v)=\zeta(t,u)+\zeta(t,v), \
\zeta(t,u)\in C[0,1]^{2}$ then (\ref{e5}) has a unique periodic
solution.
\end{ex}

\begin{proof} By Theorem \ref{thm2} it's sufficient to check
that equations (\ref{e18}) and (\ref{e19}). For $f(t,x)=f(t),\
\textrm{for all}\ x\in V$ we get\\
$$f(t)=\frac{\int_0^1\!\!\int_0^1\left(\zeta(t,u)+\zeta(t,v)\right)f(u)f(v)dudv}
{\int_0^1\!\!\int_0^1(\zeta(0,u)+\zeta(0,v))f(u)f(v)dudv}=\frac{\int_0^1
\zeta(t,u)f(u)du}{\int_0^1\zeta(0,u)f(u)du}=(Af)(t).$$\\
The equation $(Af)(t)=f(t), \ f(t)>0$ has unique a solution (see
\cite{re}). Similarly, (\ref{e19}) can be written as
$(Af)(t)=g(t), \ (Ag)(t)=f(t).$ In \cite{rh2015} it is proved that
this system of equation has not any solution in $\{(f,g)\in
(C[0,1])^{2}| \ f(t)>0,\ g(t)>0\}.$
\end{proof}

\section{An example of non-uniqueness of Gibbs measures for Hamiltonian (\ref{e1})}

 Define the operator $W:C[0,1]\to C[0,1]$ by
\begin{equation}\label{e21}(Wf)(t)=\int_0^1\!\int^1_0
K(t,u,v)f(u)f(v)dudv.
\end{equation} Then equation (\ref{e18}) can be written as
\begin{equation}\label{e22}f(t)=(Af)(t)={(Wf)(t)\over (Wf)(0)}, \ f\in
C^+[0,1]. \end{equation} Denote

$$\xi_{1}(t,u,v)=\frac{1}{\beta J_{3}}\ln\left(1+\left(t-\frac{1}{2}\right)^{\tau}
\left(u-\frac{1}{2}\right)^{\tau}\left(v-\frac{1}{2}\right)^{\tau}\left(4^{\tau}
(\tau+1)^{2}-\frac{1}{\left(v-\frac{1}{2}\right)^{\tau}+1}\right)\right),$$
where $t,u,v \in [0,1], \ \tau\in\{\frac{p}{q}\in\mathbb{Q}\ |\
p,q \ \textrm{odd positive numbers}\}$. Then, for the kernel
$K_{\tau}(t,u,v)$ of the integral operator (\ref{e22}) we have
$$K_{\tau}(t,u,v)=1+\left(t-\frac{1}{2}\right)^{\tau}
\left(u-\frac{1}{2}\right)^{\tau}\left(v-\frac{1}{2}\right)^{\tau}\left(4^{\tau}
(\tau+1)^{2}-\frac{1}{\left(v-\frac{1}{2}\right)^{\tau}+1}\right).$$\\
 Clearly, for all $t,u,v \in[0,1],$ we have
$\lim_{\tau\rightarrow 0}K_{\tau}(t,u,v)>0.$ As a result we get
following remark

\begin{rk}\label{rem5}  There exists $\tau_{0}$ such that for
every $\tau\geq \tau_{0}$ the function $K_{\tau}(t,u,v)$ is a
positive function.
\end{rk}
Put $$\Im=\left\{\frac{p}{q}\in\mathbb{Q}\ |\ p,q \ \textrm{odd
positive numbers}\right\}\bigcap \left\{\tau\in \mathbb{Q}\ | \
K_{\tau}(t,u,v)>0\right\}.
$$
\begin{pro}\label{p3.1.} For $\tau\in \Im$ the operator
$A:$
$$(Af)(t)={(Wf)(t)\over (Wf)(0)},$$ in the space $C[0,1]$ has at least two
strictly positive fixed points.
\end{pro}
\begin{proof} a) Let $f_{1}(t)\equiv1.$ Then from the equality
$$\int^1_0\!\int_{0}^1\left(u-\frac{1}{2}\right)^{\tau}\left(v-\frac{1}{2}\right)^{\tau}\left(4^{\tau}
(\tau+1)^{2}-\frac{1}{\left(v-\frac{1}{2}\right)^{\tau}+1}\right)dudv=0,$$\\
we have
$$(Af_{1})(t)=\frac{\int^1_0\!\int_{0}^1 \left[1+\left(t-\frac{1}{2}\right)^{\tau}
\left(u-\frac{1}{2}\right)^{\tau}\left(v-\frac{1}{2}\right)^{\tau}\left(4^{\tau}
(\tau+1)^{2}-\frac{1}{\left(v-\frac{1}{2}\right)^{\tau}+1}\right)\right]dudv}{\int^1_0\!\int_{0}^1
\left[1-\left(\frac{1}{2}\right)^{\tau}
\left(u-\frac{1}{2}\right)^{\tau}\left(v-\frac{1}{2}\right)^{\tau}\left(4^{\tau}
(\tau+1)^{2}-\frac{1}{\left(v-\frac{1}{2}\right)^{\tau}+1}\right)\right]dudv}=1.$$
b) Denote
$$f_{2}(t)\equiv\frac{2^{\tau}}{2^{\tau}-1}\left(1+\left(t-\frac{1}{2}\right)^{\tau}\right).$$\\
Clearly, $f_{2}\in C[0,1]$ and the function $f_{2}(t)$ is strictly
positive. Then $(Af_{2})(t)$ is equal to
$$\frac{\int^1_0\!\int_{0}^1
\left[1+\left(t-\frac{1}{2}\right)^{\tau}
\left(u-\frac{1}{2}\right)^{\tau}\left(v-\frac{1}{2}\right)^{\tau}\left(4^{\tau}
(\tau+1)^{2}-\frac{1}{\left(v-\frac{1}{2}\right)^{\tau}+1}\right)\right]
\left(1+\left(u-\frac{1}{2}\right)^{\tau}\right)\left(1+\left(v-\frac{1}{2}\right)^
{\tau}\right)dudv}{\int^1_0\!\int_{0}^1
\left[1-\left(\frac{1}{2}\right)^{\tau}
\left(u-\frac{1}{2}\right)^{\tau}\left(v-\frac{1}{2}\right)^{\tau}\left(4^{\tau}
(\tau+1)^{2}-\frac{1}{\left(v-\frac{1}{2}\right)^{\tau}+1}\right)\right]
\left(1+\left(u-\frac{1}{2}\right)^{\tau}\right)\left(1+\left(v-\frac{1}{2}\right)^{\tau}\right)dudv}.$$\\
We have

$$\int^1_0\!\int_{0}^1\left(1+\left(u-\frac{1}{2}\right)^{\tau}\right)\left(1+\left(v-\frac{1}{2}\right)^
{\tau}\right)dudv=1,$$\\ and
$$\int^1_0\!\int_{0}^1\left(u-\frac{1}{2}\right)^{\tau}\left(v-\frac{1}{2}\right)^{\tau}
\left(1+\left(u-\frac{1}{2}\right)^{\tau}\right)dudv=0.$$\\
Consequently, one gets

$$(Af_{2})(t)=\frac{1+16^{\tau}
(2\tau+1)^{2}\left(t-\frac{1}{2}\right)^{\tau}\int^1_0\!\int_{0}^1
\left(u-\frac{1}{2}\right)^{\tau}\left(v-\frac{1}{2}\right)^{\tau}
\left(1+\left(u-\frac{1}{2}\right)^{\tau}\right)\left(1+\left(v-\frac{1}{2}\right)^
{\tau}\right)dudv}{1-8^{\tau} (2\tau+1)^{2}\int^1_0\!\int_{0}^1
\left(u-\frac{1}{2}\right)^{\tau}\left(v-\frac{1}{2}\right)^{\tau}
\left(1+\left(u-\frac{1}{2}\right)^{\tau}\right)\left(1+\left(v-\frac{1}{2}\right)^
{\tau}\right)dudv}.$$\\
Since

$$\int^1_0\!\int_{0}^1\left(u-\frac{1}{2}\right)^{\tau}\left(v-\frac{1}{2}\right)^{\tau}
\left(1+\left(u-\frac{1}{2}\right)^{\tau}\right)\left(1+\left(v-\frac{1}{2}\right)^
{\tau}\right)dudv=16^{\tau}(2\tau+1)^2,$$\\
we have

$$(Af_{2})(t)=\frac{1+\left(t-0.5\right)^{\tau}}{1-0.5^\tau}=f_{2}(t).$$\\
This completes the proof.
\end{proof}
Thus, we can conclude with the following
\begin{thm}\label{thm3} Let $\sigma\in\Omega_{V}$ and $\tau\in \Im$. Then the model
$$H(\sigma)=-\frac{1}{\beta}\sum\limits_{\langle y,x,z\rangle\atop{x,y,z}\in V}\ln\left[1+\left(\sigma(x)-\frac{1}{2}\right)^{\tau}
\left(\sigma(y)-\frac{1}{2}\right)^{\tau}\left(\sigma(z)-\frac{1}{2}\right)^{\tau}\left(4^{\tau}
(\tau+1)^{2}-\frac{1}{\left(\sigma(z)-\frac{1}{2}\right)^{\tau}+1}\right)\right]$$
on the Cayley tree $\Gamma_2$ has at least two
translation-invariant Gibbs measures.
\end{thm}

Previously, it was known that for model (\ref{e1}) with
$J_{3}=J=\alpha=0,\ J_{1}\neq 0$ there exist $G^{(2)}_k$-periodic
and translation-invariant Gibbs measures it has been proved that
for some $K(t,u,v)$ (see \cite{ehr2012}, \cite{rh2015}) here exist
phase transitions (by phase transition we mean non-uniqueness of a
splitting Gibbs measure). In this section we considered
translation-invariant Gibbs measures for Hamiltonian (\ref{e1}) in
the case $J_{3}\neq 0, J=J_{1}=\alpha=0$. In other cases the problem of existence of phase transition remains open.\\

\begin{center}ACKNOWLEDGMENTS
\end{center}
The authors would like to thank Professor Yu.M.Suhov  for useful
remarks.

\end{document}